\documentclass[final,reqno,conference]{IEEEtran}
\usepackage{mathtools}

\usepackage{hyperref}
\usepackage{amsmath,amssymb,amsthm}
\usepackage{pgf, tikz}
\usetikzlibrary{arrows}

\newcommand{\eqdef}{\overset{def}{=}}

\hypersetup{
  colorlinks   = true,    
  urlcolor     = blue,    
  linkcolor    = blue,    
  citecolor    = red      
}

\newcommand{\fr}{{f_{r}}}
\newcommand{\fAC}{{f_{AC}}}
\newcommand{\fABC}{f_{ABC}}
\newcommand{\fBC}{f_{BC}}
\newcommand{\fBAC}{f_{BAC}}

\newcommand{\CAB}{{C_{AB}}}
\newcommand{\CAC}{{C_{AC}}}
\newcommand{\CABC}{C_{ABC}}
\newcommand{\CBC}{C_{BC}}
\newcommand{\CBAC}{C_{BAC}}

\newcommand{\nAC}{n_{AC}}
\newcommand{\n}{n}
\newcommand{\nABC}{n_{ABC}}
\newcommand{\nBC}{n_{BC}}
\newcommand{\nBAC}{n_{BAC}}
\newcommand{\tot}{N}
\newcommand{\C}{C}

\newtheorem{theorem}{Theorem}
\newtheorem{corrollary}[theorem]{Corollary}
\newtheorem{proposition}[theorem]{Proposition}
\newtheorem{example}[theorem]{Example}
\newtheorem{lemma}[theorem]{Lemma}

\begin{document}
\definecolor{zzttqq}{rgb}{0.6,0.2,0.0}
\definecolor{cqcqcq}{rgb}{0.7529411764705882,0.7529411764705882,0.7529411764705882}

\newcommand{\F}[1]{\textbf{F#1}}

\title{Load Balancing Congestion Games and their Asymptotic Behavior}

\author{%
\IEEEauthorblockN{Eitan Altman\IEEEauthorrefmark{1},
Corinne Touati\IEEEauthorrefmark{1}\IEEEauthorrefmark{2}%
\IEEEauthorblockA{\IEEEauthorrefmark{1}Inria\\
\IEEEauthorrefmark{2} CNRS, LIG, Univ. Grenoble Alpes, LIG, F-38000 Grenoble, France\\
Email: \{eitan.altman, corinne.touati\}@inria.fr}
}}
\maketitle

\begin{abstract}
A central question in routing games has been to establish conditions for the uniqueness of the equilibrium, either in terms of network topology or in terms of costs. This question is well understood in two classes of routing games. The first is the non-atomic routing introduced by Wardrop on 1952 in the context of road traffic in which each player (car) is infinitesimally small; a single car has a negligible impact on the congestion. Each car wishes to minimize its expected delay. Under arbitrary topology, such games are known to have a convex potential and thus a unique equilibrium. The second framework is splitable atomic games: there are finitely many players, each  controlling the route of a population of individuals (let them be cars in road traffic or packets in the communication networks). In this paper, we study two other frameworks of routing games in which each of several players has an  integer number of connections (which are population of packets) to route and where there is a constraint that a connection cannot be split. 
Through a particular game with a simple three link topology, we identify
various novel and surprising properties of games within these frameworks.
We show in particular that equilibria are non unique even in the potential
game setting of Rosenthal with strictly convex link costs. We further show
that non-symmetric equilibria arise in symmetric networks.
\end{abstract}

\section{Introduction}

A central question in routing games has been to establish conditions for
the uniqueness of the equilibria, either in terms of the network topology
or in terms of the costs.  A survey on these issues is given in~\cite{R1}.

The question of uniqueness of equilibria has been studied in two
different frameworks.  The first, which we call
\F1, is the \emph{non-atomic routing} introduced
by Wardrop on 1952 in the context of road traffic in which each player (car)
is infinitesimally small; a single car has a negligible impact on
the congestion. Each car wishes to minimize its expected delay.
Under arbitrary topology, such games are known to have a convex potential 
and thus have a unique equilibrium~\cite{R2}.
The second framework, denoted by \F2, is \emph{splitable 
atomic games}. There are finitely
many players, each  controlling the route of a population of individuals.
This type of games have already been studied in the context of road traffic
by Haurie and Marcotte \cite{RR} but have become central in the telecom
community to  model routing decisions of Internet Service Providers
that can decide how to split the traffic of their subscribers among
various routes so as to minimize network congestion~\cite{ISP}. 

In this paper we study properties 
of equilibria in two other frameworks of routing games
which exhibit surprising behavior.
The first, which we call \F3,
known as \emph{congestion games} \cite{rosenthal}, consists of atomic
players with non splitable traffic: each player has 
to decide on the path to be followed by for its traffic and cannot split
the traffic among various paths. This is a non-splitable
framework. We further introduce a new semi-splitable framework, denoted by \F4,
in which each of several players has an  integer number of connections
to route. It can choose different routes for different  connections
but there is a constraint that the traffic of a connection cannot be
split. In the case where each player controls 
the route of a single 
connection and all connections have the same size, this reduces to
the congestion game of Rosenthal~\cite{rosenthal}.

We consider in this paper routing games with 
additive costs (i.e. the cost of a path equals
to the sum of costs of the links over
the path) and the cost of a link is assumed to be convex increasing in the
total flow in the link.
The main goal of this
paper is to study a particular symmetric game of this type in a simple
topology consisting of three nodes and three links. We focus both on 
the uniqueness issue as well as on other properties of the equilibria.

This game has already been studied within the two frameworks \F1-\F2 that
we mentioned above. In both frameworks it was shown \cite{AKH} to have a
unique equilibrium. Our first finding is that in frameworks \F3 and \F4
there is a multitude of equilibria. The price of stability is thus
different than the price of anarchy and we compute both.
We show the uniqueness of the equilibrium
in the limit as the number of players $N$ grows to infinity extending 
known results \cite{RR} from framework \F2 to the new frameworks.
In framework \F2 uniqueness is in fact achieved not only for
the limiting games but also for all $N$ large enough. We show
that this {\it is not the case} for \F3-\F4:
for any finite $N$ there may be several equilibria.
We finally show a surprising property of \F4 that exhibits non
symmetric equilibria in our symmetric network example while under
\F1, \F2 and \F3 there are no asymmetric equilibria.

The structure of the paper is as follows. We first introduce the model and the notations 
used in the while study, we then move on to the properties of frameworks \F3 (Section~\ref{sec:F3})
and \F4 (Section~\ref{sec:F4}) before concluding the paper. For completeness, we also include in the
Appendix the proofs of the theorems and propositions of the paper although they will be removed from the 
final manuscript so as to comply with the conference regulations for final manuscript but will be made 
available on ArXiv.

\section{Model and Notations}
\label{sec:model}

We shall use throughout the term \emph{atomic game} to denote 
situations in which decisions of a player 
have an impact on other players' utility. It is \emph{non-atomic}
when players are  infinitesimally small and are viewed like a fluid
of players, such that a single player has a negligible impact
on the utility of other players. 

We consider a system of three nodes ($A$, $B$ and $C$) with two incoming traffic sources (respectively from node $A$ and $B$) and an exit node $C$. There are a total of $\tot$ connections originating from each one of the sources. Each connection can either be sent directly to node $C$ or rerouted via the remaining node.
The system is illustrated in Figure~\ref{fig:system}.

\definecolor{zzttqq}{rgb}{0.6,0.2,0.0}
\definecolor{qqqqff}{rgb}{0.0,0.0,1.0}
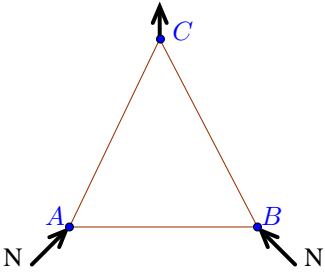
\begin{figure}[htb]
\begin{center}
\begin{tikzpicture}[line cap=round,line join=round,>=angle 45,x=1.0cm,y=1.0cm]
\draw [color=zzttqq] (2.2,4)-- (1.0,1.5); \draw [color=zzttqq] (1.0,1.5)-- (3.5,1.5);
\draw [color=zzttqq] (3.5,1.5)-- (2.2,4);
\draw [->,line width=1.5pt] (0.5,1.0) -- (1.0,1.5);
\draw [->,line width=1.5pt] (4.0,1.0) -- (3.5,1.5);
\draw [->,line width=1.5pt] (2.2,4) -- (2.2,4.5);
\draw (0,1.34) node[anchor=north west] {N};
\draw (4,1.34) node[anchor=north west] {N};
\draw [fill=qqqqff] (2.21,4) circle (1.5pt);
\draw[color=qqqqff] (2.51,4.11) node {$C$};
\draw [fill=qqqqff] (1.0,1.5) circle (1.5pt);
\draw[color=qqqqff] (0.81,1.65) node {$A$};
\draw [fill=qqqqff] (3.5,1.5) circle (1.5pt);
\draw[color=qqqqff] (3.69,1.65) node {$B$};
\end{tikzpicture}
\end{center}
\caption{Physical System}
\label{fig:system}
\end{figure}

This model has been used to model load balancing issues in computer networks, see \cite{AKH} and references therein.
Jobs arrive to two computing centers represented by nodes A and B. A job can be processed locally at the node where
it arrives or it may be forwarded to the other node incurring further communication delay. The costs of links $[AC]$ and
$[BC]$ represent the processing delays of jobs processed at nodes A and B respectively. Once processed, the jobs leave
the system. A connection is a collection of jobs with similar characteristics (e.g. belonging to the same application). 

We introduce the following notations:
\begin{itemize}
\item A link between two nodes, say A and B, is denoted by $[AB]$. Our considered system
has three links $[AB]$, $[BC]$ and $[AC]$.
\item A route is simply referred by a sequence of nodes. Hence, the system has four connections: two originating from node $A$ (route AC and ABC) and two originating from node $B$ (route BC and BAC).
\end{itemize}

Further, in the following, $\nAC$, $\nBC$, $\nABC$ and $\nBAC$ will refer to the number of connections routed via the different routes while $\n[AC]$, $\n[BC]$ and $\n[AB]$ will refer to the number of connections on each subsequent link.
By conservation law, we have:
$$ \nAC+\nABC = \nBC + \nBAC = \tot$$
$$ \text{and }\left\{
\begin{array}{l}
\n[AC] = \nAC + \nBAC, \\ \n[BC] = \nABC + \nBC, \\ 
\n[AB] = \nBAC + \nABC. 
\end{array} \right.
$$

For each route $r$, we also define the fraction (among $N$)
of flow using it, i.e. $\fr = \n_r / N$. The conservation law becomes $\fAC+\fABC = \fBC + \fBAC = 1$.

Finally, the performance measure considered in this work is the cost (delay) of connections experienced on their route. We consider a simple model in which the cost is additive (i.e. the cost of a connection on a route is simply taken as the sum of delays experienced by the connection over the links that constitute this route). We further assume that the costs on each link are linear with coefficient $a/N$ on link $[AB]$ and coefficient $b/N$ on link $[AC]$ and $[BC]$, i.e.
$$ \left\{
\begin{array}{ll}
\displaystyle \C_{[AB]} = \frac{a}{N} \n[AB] = a (\fBAC+\fABC),\\
\displaystyle \C_{[AC]} = \frac{b}{N} \n[AC] = b (\fBAC+\fAC),\\
\displaystyle \C_{[BC]} = \frac{b}{N} \n[BC] = b (\fBC+\fABC). 
\end{array}\right.
$$
and then:
$$ \begin{array}{@{}ll@{}}
\CAB = \C_{[AB]}, &
\CABC = \C_{[AB]} + \C_{[BC]}, \\
\CBC = \C_{[BC]}, & 
\CBAC = \C_{[AB]} + \C_{[AC]}.
\end{array}
$$
We restrict our study to the (pure) Nash equilibria and give the equilibria in terms of the corresponding flows marked by a star. By conservation law, the equilibria is uniquely determined by the specification of $\fABC^*$ and $\fBAC^*$ (or equivalently 
$\nABC^*$ and $\nBAC^*$).

We recall that in this paper, we consider two types of decision models.
In the first (\F3), the decision is taken at the connection level 
(Section~\ref{sec:F3}), i.e. each connection has its own decision maker
that seeks to minimize the connection's cost, and the connection cannot be split into different routes.
In the second (\F4), (Section~\ref{sec:F4}) each one of the two source nodes decides on the routing of all
the connections originating there. Each connection of a given source node (either A or B) can be routed independently 
but a connection cannot be split into different route. We hence refer to \F4 this semi-splitable framework.
Note that the two-approaches (\F3 and \F4) coincide when there is only $N=1$ connection at each source, which we also detail later.

\section{Atomic Non-Splitable Case and its non-atomic limit (\F3 framework)}
\label{sec:F3}

We consider here the case where each connection belongs to an
 individual user acting selfishly. 

We first show that for fixed parameters,
the game may have several equilibria, all of
which are symmetric for any number of players.
The  number of distinct equilibria can be made arbitrary large
by an appropriate choice of the parameters $a$ and $b$, and for any choice of $a$ and $b$, there exists $N_0$ such that the number of equilibria remain constant for all $N \geq N_0$.
We then show properties of
the limiting game obtained as the number of 
of players increases to infinity.

\subsection{Non-uniqueness of the equilibrium}

\begin{theorem}
The set of pure Nash equilibria of the game are the points satisfying
$\displaystyle \nBAC^* = \nABC^* \leq \frac{b}{2a}$.
\label{th:nonSplit}
\end{theorem}

\begin{proof}
Consider an equilibrium $(\nABC^*, \nBAC^*)$. Then, we have the following conditions:
\begin{equation} \left\{
\begin{array}{@{}l}
\C_{[AC]}  = \CAC \leq (\C_{[AB]}+a/N) + (\C_{[BC]}+b/N) \\
\C_{[BC]} = \CBC \leq (\C_{[AB]}+a/N) + (\C_{[AC]}+b/N) \\
\C_{[AB]} + \C_{[BC]} = \CABC \leq \C_{[AC]}+b/N \\ \label{eq:th1}
\C_{[AB]} + \C_{[AC]} = \CBAC \leq \C_{[BC]}+b/N 
\end{array}\right.\end{equation}


Note that the last two equations lead to:
$$
\left\{
\begin{array}{l}
\C_{[AB]} \leq - \C_{[BC]} + \C_{[AC]}+b/N \\
\C_{[AB]} \leq - \C_{[AC]} + \C_{[BC]}+b/N 
\end{array}\right.$$

One can check that $(\nABC^*, \nBAC^*) =$ $(0,0)$ is a solution. If the equilibrium is not the trivial null solution, then either $\nABC^* \neq 0$ or $\nBAC^* \neq 0$. Either way leads to $\C_{[AB]} > 0$ and thus $-b/N < \C_{[AC]}-\C_{[BC]} < b/N$ which implies that $\C_{[AC]}=\C_{[BC]}$. Equation~\ref{eq:th1} becomes:
$$ \left\{
\begin{array}{@{}l}
0 \leq a (\nABC^* + 1 + \nBAC^*) + b\\
a (\nABC^* + \nBAC^*) \leq b 
\end{array}\right. \Leftrightarrow
a (\nABC^* + \nBAC^*) \leq b
$$

But then:\\
$\C_{[AC]}=\C_{[BC]} \Leftrightarrow b (\nAC^*+\nBAC^*) =  b (\nBC^* + \nABC^*) \Leftrightarrow$ 
$N-\nABC^*+\nBAC^* = N - \nBAC^* + \nABC^* \Leftrightarrow \nABC^* = \nBAC^*$. Therefore the equilibrium is symmetrical.
Jointly with $a (\nABC^* + \nBAC^*) \leq b$, this leads to the conclusion.
\end{proof}

\begin{corrollary}
\label{cor2}
For $N \geq N_0 = \lceil \frac{b}{2a} \rceil$, there exists exactly $b/2a+1$ Nash equilibria in pure strategies.
\end{corrollary}

\subsection{The potential and asymptotic uniqueness}

When the number of players $N$ grows to infinity, the limiting
game becomes a non-atomic game with a potential~\cite{sandholm}
$$ F_{\infty} (\fABC, \fBAC) = b (\fABC-\fBAC)^2
+\frac{a}{2}\left(\fABC+\fBAC\right)^2 \hspace{-3pt}.$$
Indeed, recall that the potential $g$ is unique up to an additive constant 
and that it satisfies
$$ \left\{\begin{array}{@{}l@{}l@{}l@{}}
\frac{\partial g}{\partial \fAC} &\eqdef \C_{AC} = b (\fAC+\fBAC)\\
\frac{\partial g}{\partial \fABC} &\eqdef \C_{ABC} = a (\fABC+\fBAC) + b (\fABC + \fBC)\\
\frac{\partial g}{\partial \fBC} &\eqdef \C_{BC} = b (\fBC+\fABC)\\
\frac{\partial g}{\partial \fBAC} &\eqdef \C_{BAC} = a (\fABC+\fBAC) + b (\fBAC + \fAC).
\end{array} \right.$$
One can check that the function 
$$ \begin{array}{@{}l@{}l@{}}
g(&\fAC, \fABC, \fBC, \fBAC) = \frac{a}{2} (\fABC+\fBAC)^2  \\
&+ \frac{b}{2} ((\fAC+\fBAC)^2 + (\fBC+\fABC)^2) 
\end{array}
$$
readily satisfies these conditions. Then $g$ can be rewritten as
$$ \begin{array}{@{}l@{}l@{}}
g(\fABC, \fBAC) = \\
\frac{a}{2} (\fABC+\fBAC)^2 + \frac{b}{2} (1+(\fABC-\fBAC)^2). 
\end{array}
$$
As the potential is unique up to an additive constant, we consider $F_{\infty} = g-b.Id/2$.

\begin{proposition}
The non-atomic game has a unique Nash equilibrium, which is $\fABC^*=\fBAC^*=0$. 
\end{proposition}

\begin{proof}
Note that: 
$$ \left\{\begin{array}{@{}l@{}l@{}}
\displaystyle \frac{\partial F_\infty}{\partial \fABC} &= a(\fABC+\fBAC) + 2b (\fABC -\fBAC)) \\[1em]
\displaystyle \frac{\partial F_\infty}{\partial \fBAC} &= a(\fABC+\fBAC) + 2b (\fBAC -\fABC))
\end{array} \right.
$$

Hence, the potential is twice differentiable with Hessian matrix
$$ \left(
\begin{array}{cc}
a+2b & a-2b \\
a-2b & a+2b
\end{array}
 \right).$$

This Hessian is definite positive and hence the potential is (strictly) convex. Therefore it has a unique minimum, which is the only Nash equilibrium of the game. Finally, note that $\forall \fABC \in (0,1)$, $\fBAC \in (0,1)$, $F_\infty (\fABC, \fBAC) \geq 0$ and that $F_\infty (0,0) = 0$, which concludes the proof.
\end{proof} 

To show the uniqueness of the equilibrium in the limiting game, we made use of
the fact that the limiting game has a potential which is
convex. Yet, not only the limiting game has a convex potential, but also
the original one, as we conclude from next
theorem, whose proof is a direct application of \cite{rosenthal}.
\begin{theorem}
For any finite number of players, the
game is a potential game~\cite{monderer} with the potential function:
\begin{equation}\label{eq:pot}
\begin{array}{@{}l@{}l@{}} 
F(\fABC&, \fBAC) = \\ 
& bN (\fABC-\fBAC)^2\\
& \displaystyle +\frac{aN}{2}\left(\fABC+\fBAC\right)\left(\fABC+\fBAC+1/N\right).
\end{array}
\end{equation}
\end{theorem}
\begin{proof}
Consider a connection following route $ABC$. Its cost is
$ a (\fABC + \fBAC) + b (\fABC+1-\fBAC).$
If this connection switches its strategy to route $AC$, then its cost becomes
$ b (1-\fABC+\fBAC+1/N) $.
Therefore the associated change of cost is
$$ \begin{array}{@{}l@{}l@{}} 
\Delta &= a (\fABC + \fBAC) + b (\fABC+1-\fBAC) \\
& \hfill - b (1-\fABC+\fBAC+1/N) \\
&= a (\fABC + \fBAC) + b (2\fABC-2\fBAC-1/N).\end{array}$$
Now:
$$ \begin{array}{l@{}l}
\frac{1}{N} &(F(\fABC, \fBAC) - F(\fABC-1/N, \fBAC)) \\
&= b \left[(\fABC-\fBAC)^2 - (\fABC-1/N-\fBAC)^2 \right] \\
& \qquad \qquad +\frac{a}{2}\left[\left(\fABC+\fBAC\right)\left(\fABC+\fBAC+1/N\right) \right. \\
& \hfill \left.- \left(\fABC+\fBAC-1/N\right)\left(\fABC+\fBAC\right)\right]\\
& = \displaystyle \frac{b}{N} \left( 2\fABC-2\fBAC - 1/N \right)  + \frac{a}{N} \left( \fABC+\fBAC\right)\\
& = \Delta/N.
\end{array}
$$
By symmetry, the same argument holds for a connection originating from source B.
\end{proof}
Note that unlike the framework of non-atomic games, the fact that the
game has a convex potential does not imply uniqueness. The reason for that
is that in congestion games,
the action space over which the potential is minimized is not
a convex set (due to the non-splitable nature) so that 
it may have several local minima, each corresponding to another
equilibrium, whereas a for a convex function over the Euclidean
space, there is a unique local minimum which is also a global minimum
of the function (and thus an equilibrium of the game).

\subsection{Efficiency}

\begin{theorem}
In the non-atomic setting, the only Nash equilibrium is also the social optimum (i.e. the point minimizing the sum of costs of all players) of the system.
\end{theorem}
\begin{proof}
The sum of costs of all players is
\begin{equation}
\begin{array}{@{}l@{}l@{}}
\displaystyle &\fABC \CABC + \fAC \CAC + \fBAC \CBAC + \fBC \CBC  \\
&= a (\fABC+\fBAC)^2 \\
& \hfill + b ( (\fBC+\fABC)^2 + (\fAC+\fBAC)^2) \\
& = a(\fABC+\fBAC)^2 + 2b (1 + (\fABC -\fBAC)^2).
\end{array}\label{eq:sumCost}
\end{equation}

The minimum is hence obtained for $(\fABC,\fBAC) = (0,0)$.
\end{proof}

Since the game possesses several equilibria, we can expect
the PoA (Price of Anarchy - the largest ratio between the
sum of costs at an equilibrium and the sum of costs at the social 
optimum) and PoS (Price of Stability - the smallest corresponding
ratio) to be different. 

\begin{theorem}
The price of stability of the game is $1$ and the price of anarchy is $1+\frac{b}{2aN^2}$.
\end{theorem}

\begin{proof}
From Eq.~\ref{eq:sumCost} the price of anarchy (resp. stability) is by definition the maximum (resp. minimum) value over the Nash equilibria of:
$$
\frac{a(\fABC^*+\fBAC^*)^2 + 2b (1 + (\fABC^* -\fBAC^*)^2)}{2b}
$$
Then, from Theorem~\ref{th:nonSplit}:
$$ \begin{array}{@{}ll@{}}
PoA & \displaystyle = \max_{p \leq b/2a} \frac{(2p/N)^2 + 2b}{2b}
= \max_{p \leq b/2a} \frac{2a p ^2/N^2 + b}{b} \\[1.1em]
& \displaystyle =  \frac{2a (b/2aN)^2 + b}{b} =
\frac{2ab^2}{4ba^2N^2} + 1 = \frac{b}{2aN^2} + 1 
\end{array}
$$
and 
$$
PoS = \hspace{-2pt} \min_{p \leq b/2a} \hspace{-1pt} \frac{a(2p/N)^2 + 2b}{2b}
= \hspace{-1pt} \min_{p \leq b/2a} \frac{2a p ^2/N^2}{b}+1
=  1.
$$

\end{proof}

We make the following observations: \\
(i) In the splitable atomic games
studied in \cite{AKH} the PoA was shown to be greater than one 
for sufficiently small number of players (smaller than some threshold), 
and was 1 for all large enough number of players (larger than the same
threshold). Here for any number of players, the PoS is 1 and the PoA
is greater than 1. 
\\
(ii) The PoA decreases in $N$ and tends to 1 as $N$ tends to infinity,
the case of splitable games. 
\\
(iii) We have shown that the PoA is unbounded: for any real value $K$ and any number
of players one can choose the cost parameters $a$ and $b$ so that
the PoA exceeds $K$. This corresponds to what was observed in
splitable games~\cite{AKH} and contrast with the non-atomic 
setting of single commodity flows
(i.e. when there is only one source node instead of two), and arbitrary 
topology networks where the PoA equals 4/3~\cite{tim}.

\section{Atomic Semi-Splitable Case and its Splitable limit (\F4 framework)}
\label{sec:F4}

The game can be expressed as a $2$-player matrix game where each player 
(i.e. each source node $A$ and $B$) has $N+1$ possible actions, for each of the $N+1$ possible values of $\fABC$ and $\fBAC$ respectively.

The utility for player $A$ is 
\begin{equation}\label{eq:utA}
\begin{array}{@{}l@{}l@{}}
U_A(&\fABC, \fBAC) = \fAC \CAC + \fABC \CABC \\
& = b - b\fABC + b \fBAC  \\ & \quad + (a-2b) \fABC\fBAC +  (a+2b) \fABC^2
\end{array}\end{equation}

Similarly, for player $B$:
\begin{equation}\label{eq:utB}
\begin{array}{@{}l@{}l@{}}
U_B(&\fABC, \fBAC) = \fBC \CBC + \fBAC \CBAC \\
&= b -b \fBAC + b\fABC \\ & \quad + (a-2b)\fBAC\fABC + (a+2b)\fBAC^2
\end{array}\end{equation}

Note that $$\frac{\partial U_A}{\partial \fABC} = -b + (a-2b) \fBAC + 2 (a+2b) \fABC $$
 $$ \text{and } \frac{\partial U_B}{\partial \fBAC} = -b + (a-2b) \fABC + 2 (a+2b) \fBAC.$$
Hence $\displaystyle \frac{\partial^2 U_A}{\partial \fABC^2} = 2 (a+2b) =
 \frac{\partial^2 U_B}{\partial \fBAC^2}$.
Therefore, both $u_A: \fABC \mapsto U_A(\fABC, \fBAC)$
 and $u_B: \fBAC \mapsto U_B(\fABC, \fBAC)$ are (strictly) convex functions. 
This means that for each action of one player, there would be a unique best response to the second player
if its action space was the interval $(0,1)$. Hence, for the limit case (when $N \rightarrow \infty$), the
best response is unique. In contrast, for any finite value of $N$, there are either $1$ or $2$ possible best 
responses which are the discrete optima of functions $u_A:\fABC \mapsto U_A(\fABC, \fBAC)$
 and $u_B:\fBAC \mapsto U_B(\fABC, \fBAC)$. We will however show that in the finite case, there may be up to $2 \times 2 = 4$ Nash equilibria while in the limit case the equilibrium is always unique.

\subsection{Efficiency}
Note that the total cost of the players is 
$$
\begin{array}{@{}l@{}l}
\Sigma(\fABC, \fBAC) = U_A(\fABC, \fBAC) + U_B(\fABC, \fBAC) \\
= 2b + 2 (a-2b) \fABC \fBAC + (a+2b) ( \fABC^2 + \fBAC^2 ) \\
= 2b + a (\fABC+\fBAC)^2+2b (\fABC-\fBAC)^2 \\
\geq 2b.
\end{array}
$$

Further, note that $ \Sigma = 2 (F_\infty+b)$. Hence $\Sigma$ is strictly convex. Also $\Sigma(0, 0) =2b$. Therefore $(0,0)$ is the (unique) social optimum of the system. Yet, for sufficiently large $N$ (that is, as soon as we add enough flexibility in the players' strategies), this is not a Nash equilibrium, as stated in the following theorem:

\begin{theorem}
The point $(\fABC, \fBAC) = (0,0)$ is a Nash equilibrium if and only if $N \leq \frac{a}{b} + 2$. 
\end{theorem}

\begin{proof}
By symmetry and as $u_A :\fABC \mapsto U_A(\fABC, \fBAC)$ is convex, then $(0,0)$ is a Nash equilibrium iff $U_A(0,0) \leq U_A(1/N,0) = b-b/N+(a+2b)/N^2 $ which leads to the conclusion.
\end{proof}

Also, we can bound the total cost by:
$$
\begin{array}{@{}l@{}l}
\Sigma(\fABC, \fBAC) = \\ 
= 2b + 2 (a-2b) \fABC \fBAC + (a+2b) ( \fABC^2 + \fBAC^2 ) \\
\leq  2b + (a-2b) (\fABC^2 +  \fBAC^2) + (a+2b) ( \fABC^2 + \fBAC^2 ) \\
\leq 2b + 2a (\fABC^2 +  \fBAC^2) \\
\leq 2b + 4a
\end{array}
$$

This bound is attained at $\Sigma(1,1) = 2b + 2 (a-2b) + 2 (a+2b) = 4a + 2b$. Yet, it is not obtained at the Nash equilibrium for sufficiently large values of $N$:

\begin{theorem}
$(1,1)$ is a Nash equilibrium if and only if $\displaystyle N \leq \frac{2b+a}{3a+b}$.
\end{theorem}
\begin{proof}
We have $ U_A(1,1) = b + 2a$ and 
$$ U_A(1-1/N, 1) = 2a+b-3a/N-b/N+2b/N^2+a/N^2.$$
Therefore $ U_A (1-1/N, 1) \geq U_A(1,1) \Leftrightarrow 2b+a \geq (3a+b) N$.
The conclusion follows by convexity.
\end{proof}

Therefore, for $N \geq \max (\frac{a}{b} + 2, \frac{2b+a}{3a+b})$ the Nash equilibria are neither optimal nor worse-case strategies of the game.

\subsection{Case of \texorpdfstring{$N=1$}{N=1}}
\label{sec:2x2}
In case of $N=1$ (one flow arrives at each source node and there
are thus two players)
the two approach coincides: the atomic non-splitable case (\F3) is also a semi-splitable
atomic game (\F4). 
$\fABC$ and $\fBAC$ take values in $\{\{0\},\{1\}\}$. From Eq. \ref{eq:utA} and Eq. \ref{eq:utB}, the matrix game can be written
$$ \left( \begin{array}{r@{\,,\,}lr@{\,,\,}l}
(b& b) & (2b& a+2b) \\
(a+2b& 2b) & (2a+b& 2a+b)
\end{array} \right)$$
and the potential of Eq.~\ref{eq:pot} becomes
$$ \left( \begin{array}{cc}
0 & a+b \\
a+b & 3a
\end{array} \right).$$

Then, assuming that either $a$ or $b$ is non null, we get that $(0,0)$ is always a Nash equilibrium and that $(1,1)$ is a Nash equilibrium if and only if $3a \leq a+b$, i.e. $2a<b$.


We  next consider  any integer $N$ and identify another surprising feature of the equilibrium.
We show that depending on the sign of $a-2b$, non-symmetric
equilibria arise in our symmetric game. In all frameworks other
than the semi-splitable games there are only symmetric equilibria
in this game. We shall show however that in the limit (as $N$
grows to infinity), the limiting game has a single equilibrium.

\subsection{Case \texorpdfstring{$a-2b<0$}{a-2b<0}}

In this case, there may be multiple equilibria, as shown in the following example.
\begin{example}
Consider $a=1$, $b=3$ and $N=4$, then the cost matrices are given below, with the two Nash equilibria of the game represented in bold letters:
$$ \displaystyle U_A = \frac{1}{16}\left(\begin{array}{lllll}
48 & 60 & 72 & 84 & 96 \\
43 & \mathcal{\mathbf{50}} & 57 & 64 & 71\\
52 & 54 & \mathcal{\mathbf{56}} & 58 & 60\\
75 & 72 & 69 & 66 & 63\\
112 & 104 & 96 & 88 & 80\\ 
\end{array}
\right) \hspace{-2pt}, \, \text{and} $$
$$ U_B =
\frac{1}{16}\left(\begin{array}{lllll}
48 & 43 & 52 & 75 & 112 \\
60 & \mathcal{\mathbf{50}} & 54 & 72 & 104\\
72 & 57 & \mathcal{\mathbf{56}} & 69 & 96\\
84 & 64 & 58 & 66 & 88\\
96 & 71 & 60 & 63 & 80\\ 
\end{array} \right)\hspace{-2pt}.
$$
\end{example}

Note that due to the shape of $U_A$ and $U_B$ the cost matrices of the game are transpose of each other. Therefore in the following, we shall only give matrix $U_A$.

We have the following theorem:
\begin{theorem}\label{th:asmall}
All Nash equilibria are symmetrical, i.e.
$$\fABC^* = \fBAC^*. $$
\end{theorem}
The proof is given in Appendix~\ref{ap:asmall}.

\subsection{Case \texorpdfstring{$a=2b$ (with $a>0$)}{a=2b}}

When $a=2b$, we shall show that some non-symmetrical equilibria exists.

\begin{theorem} \label{th:a=2b}
If $a=2b$, there are exactly either $1$ or $4$ Nash equilibria. For any $N$, let $\overline{N} = \lfloor \frac{N}{8} \rfloor$.
\begin{itemize}
\item If $N \text{ mod } 8 =4$, there are $4$ equilibria $(\nABC^*,\nBAC^*)$, which are $(\overline{N}, \overline{N})$, $(\overline{N}+1, \overline{N})$, $(\overline{N}, \overline{N}+1)$ and $(\overline{N}+1, \overline{N}+1)$.
\item Otherwise, there is a unique equilibrium, which is $(\overline{N}, \overline{N})$ if $N\text{ mod } 8 < 4$ or 
$(\overline{N}+1, \overline{N}+1)$ if $N \text{ mod } 8 > 4$.
\end{itemize}
\end{theorem}

\begin{proof}
The Nash equilibria are the optimal points for both $u_A$ and $u_B$. They are therefore either interior or boundary points (i.e. either $\fABC$ or $\fBAC$ are in ${{0},{1}}$). We detail the interior point cases in Appendix~\ref{ap:a=2b}.
The rest of the proof derives directly from the definition of $\displaystyle \frac{\partial U_A}{\partial \fABC}$ and $\displaystyle \frac{\partial U_B}{\partial \fBAC}$. Indeed:
$$ \frac{\partial U_A}{\partial \fABC} = (a-2b) \fBAC + 2 (2b+a) \fABC -b = 8b \fABC -b $$
$$ \frac{\partial U_B}{\partial \fBAC} = (a-2b) \fABC + 2 (a+2b) \fBAC -b = 8b \fBAC -b.$$
Both are minimum for $1/8$. Therefore, it is attained if $N$ is a multiple of $8$. Otherwise, the best response of each player is either $\frac{1}{N}\lfloor \frac{N}{8} \rfloor$ if $N \mod 8 \leq 3$ or $\frac{1}{N} \lceil \frac{N}{8} \rceil$ if $N \mod 8 \geq 5$. If $N \mod 8 = 4$, then each player has $2$ best responses which are $\frac{1}{N} \frac{N-4}{8}$ and $\frac{1}{N} \frac{N+4}{8}$. Then, one can check that the boundary points follow the law of Theorem~\ref{thB:a=2b} when $\overline{N} = \lfloor \frac{N}{8} \rfloor = 0$.
\end{proof}

\subsection{Case \texorpdfstring{$a-2b>0$}{a-2b>0}}


\begin{theorem} \label{th:big}
If $a-2b>0$, there are exactly either $1$, $2$ or $3$ Nash equilibria.\\ 
Let $\displaystyle \alpha = \frac{a+2b}{3a+2b}$, $\displaystyle \beta=\frac{2a}{3a+2b}$ and $\displaystyle \gamma = \frac{b}{3a+2b}$.\\[0.9em]
Define further $\widetilde{N} = \lfloor N\gamma \rfloor$ and $z(N)=N\gamma-\widetilde{N}$.
The equilibria are of the form 
\begin{itemize}
\item Either $(\widetilde{N}, \widetilde{N})$, $(\widetilde{N}+1, \widetilde{N})$, $(\widetilde{N}, \widetilde{N}+1)$\\ if $N$ is such that $z(N) = \alpha$ (mode 3-A in Figure~\ref{fig:modes})
\item Or $(\widetilde{N}+1, \widetilde{N}+1)$, $(\widetilde{N}+1, \widetilde{N})$, $(\widetilde{N}, \widetilde{N}+1)$ if $N$ is such that $z(N) = \beta$ (mode 3-B)
\item Or $(\widetilde{N}, \widetilde{N}+1)$, $(\widetilde{N}+1, \widetilde{N})$\\ if $N$ is such that $\alpha < z(N) < \beta$ (mode 2)
\item Or $(\widetilde{N}, \widetilde{N})$\\ if $N$ is such that $\beta < z(N) < \alpha+1$ (mode 1).
\end{itemize}\vspace{-1em}
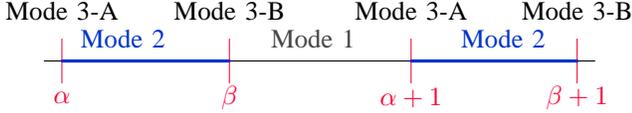
\begin{figure}[htb]
\definecolor{uuuuuu}{rgb}{0.27,0.27,0.27} 
\definecolor{qqttcc}{rgb}{0.0,0.2,0.8}
\definecolor{ffqqtt}{rgb}{1.0,0.0,0.2}
\begin{tikzpicture}[line cap=round,line join=round,>=triangle 45,x=0.9433962264150945cm,y=0.9771986970684041cm]
\newcommand\XA{3.24} \newcommand\XB{5.6} \newcommand\XC{8.16} \newcommand\XD{10.50}
\newcommand\haut{4.3}\newcommand\bas{3.7}\newcommand\centre{4}
\clip(\XA-1,3) rectangle (\XD+1,5.5);
\draw [domain=\XA-0.25:\XD+0.25] plot(\x,{\centre}); 
\draw [color=ffqqtt] (\XA,\haut)-- (\XA,\bas); \draw [color=ffqqtt] (\XB,\haut)-- (\XB,\bas); 
\draw [color=ffqqtt] (\XC,\haut)-- (\XC,\bas); \draw [color=ffqqtt] (\XD,\haut)-- (\XD,\bas);
\draw [line width=1.20pt,color=qqttcc] (\XA,\centre)-- (\XB,\centre); 
\draw [line width=1.20pt,color=qqttcc] (\XC,\centre)-- (\XD,\centre);
\draw (\XA,4.94) node[anchor=north] {Mode 3-A}; 
\draw (\XB,4.94) node[anchor=north] {Mode 3-B};
\draw (\XC,4.94) node[anchor=north] {Mode 3-A};
\draw (\XD,4.94) node[anchor=north] {Mode 3-B};
\draw[color=ffqqtt] (\XA,3.5) node {$\alpha$};
\draw[color=ffqqtt] (\XB,3.5) node {$\beta$};
\draw[color=ffqqtt] (\XC,3.5) node {$\alpha+1$};
\draw[color=ffqqtt] (\XD,3.5) node {$\beta+1$};
\draw[color=qqttcc] (4.10,\haut) node {Mode 2};
\draw[color=qqttcc] (9.46,\haut) node {Mode 2};
\draw[color=uuuuuu] (6.78,\haut) node {Mode 1};
\end{tikzpicture}
\caption{Different modes according to different values of $N$.}
\label{fig:modes}
\end{figure}
\end{theorem}

We illustrate the different modes in the following example.
\begin{example}
Suppose that $a=10$ and $b=3$ (we represent only the part of the matrices corresponding to $1/N \leq \fABC, \fBAC \leq 4/N$). 

If $N=24$, there are $3$ Nash equilibria:
$$\begin{array}{llll}
1152&	1200&	1248&	1296\\
1118&	\mathbf{1172}&	\mathbf{1226}&	1280\\
1112&	\mathbf{1172}&	1232&	1292\\
1134&	1200&	1266&	1332\\
\end{array}$$

If $N=26$, there are $2$ Nash equilibria:
$$\begin{array}{llll}
1352&	1404&	1456&	1508\\
1314&	1372&	\mathbf{1430}&	1488\\
1304&	\mathbf{1368}&	1432&	1496\\
1322&	1392&	1462&	1532\\
\end{array}$$

If $N=27$, there are $3$ Nash equilibria:
$$\begin{array}{llll}
1458&	1512&	1566&	1620\\
1418&	1478&	\mathbf{1538}&	1598\\
1406&	\mathbf{1472}&	\mathbf{1538}&	1604\\
1422&	1494&	1566&	1638\\
\end{array}$$

If $N=28$, there is a single Nash equilibrium:
$$\begin{array}{llll}
1568&	1624&	1680&	1736\\
1526&	1588&	1650&	1712\\
1512&	1580&	\mathbf{1648}&	1716\\
1526&	1600&	1674&	1748\\
\end{array}$$
\end{example}

\subsection{Limit Case: Perfectly Splitable Sessions}

We focus here in the limit case where $N \rightarrow +\infty$. 
\begin{theorem}
There exists a unique Nash equilibrium and it is such that $$ \fBAC^* = \fABC^* = \frac{b}{3a+2b}.$$
\end{theorem}

\begin{proof}
Note that $\displaystyle \frac{\partial U_A}{\partial \fABC} (1) >0$ and $\displaystyle \frac{\partial U_B}{\partial \fBAC} (1) >0$.
If $\fABC = 0$ then $\fBAC = \frac{b}{2a+4b}$ which implies that $ -b+ \frac{b (a-2b)}{2a+4b} \geq 0$, which further implies that $-a-6b >0$ which is impossible. Hence $\fABC >0$. Similarly $\fBAC >0$ which concludes the proof.
\end{proof}

Recall that the optimum sum (social optimum) is given by $(0,0)$ and that the worse case is given by $(1,1)$. Hence, regardless of the values of $a$ and $b$, at the limit case, we observe that there is a unique Nash equilibrium, that is symmetrical, and is neither optimal (as opposed to \F3), nor the worst case scenario. 
The price of anarchy is then:
$$ PoA = PoS = \frac {2b + 2 \fABC^{*^2}a}{2b} =  1 + \frac{ab}{(3a+2b)^2}.$$

\section{Conclusions}
We revisited in this paper a load balancing problem within a non-cooperative
routing game framework. This model had already received
much attention in the past within some classical frameworks (the Wardrop
equilibrium analysis and the atomic splitable routing game framework).
We studied this game under other frameworks - the non splitable atomic game
(known as congestion game) as 
well as a the semi-splitable framework.
We have identified many surprising features of
equilibria in both frameworks. We showed that unlike the previously
studied frameworks, 
there is no uniqueness of equilibrium, and non-symmetric equilibria may
appear (depending on the parameters). For each of the frameworks we identified the
different equilibria and provided some of their properties. We also provided
an efficiency analysis in terms of price of anarchy and price of stability.
In the future we plan to investigate more general cost structures and topologies.

\bibliographystyle{IEEEtran}
\bibliography{routing}

\appendix

\subsection{Proof of Theorem~\ref{th:asmall}.}
\label{ap:asmall}

Suppose that $(\fABC^*,\fBAC^*)$ is a Nash equilibrium with $\fABC^* \neq \fBAC^*$. Then, by definition:
$$ \begin{array}{l}
U_A (\fABC^*, \fBAC^*) \leq U_A (\fBAC^*, \fBAC^*) \text { and } \\
U_B (\fABC^*, \fBAC^*) \leq U_B (\fABC^*, \fABC^*),
\end{array}
$$
which gives, after some manipulations,
$$ \left\{ \begin{array}{l@{}l@{}} 
(a-2b)&\fABC^*\fBAC^* \leq \\
&2a \fBAC^{*2} + b\fABC^* - b \fBAC - (a+2b) \fABC^{*2} \\
(a-2b)&\fABC^*\fBAC^*  \leq \\
&2 a \fABC^{*2} + b \fBAC^* - b\fABC^* - (a+2b)\fBAC^{*2}.
\end{array}\right.$$

Therefore $ \displaystyle 2 (a-2b) \fABC^*\fBAC^* \leq (a-2b) ( \fABC^{*2}  + \fBAC^{*2})$
and hence $\displaystyle 0 \leq (a-2b) ( \fABC^* - \fBAC^*)^2 $ which is impossible.

\subsection{Boundary equilibria when \texorpdfstring{$a=2b$}{a=2b}.}
\label{ap:a=2b}

\begin{theorem} \label{thB:a=2b} 
If $a=2b$, there exists a single Nash equilibrium of the form $(0, \fBAC^*)$ and $(\fBAC^*, 0)$ with $\fBAC^*$ non null. It is obtained for $N=4$ and $\fBAC^* = 1/4$.
 The points $(0, 0)$ are Nash equilibria if and only if $N \leq 4$.
 Further, there are no equilibrium of the form $(\fABC, 1)$ or $(1, \fBAC)$.
 \end{theorem}
\begin{proof}
We first study the equilibria of the form $(0, \fABC)$.
$(0, \gamma)$ is a Nash equilibrium iff
$$ \left\{\begin{array}{@{\,}l@{}}
\displaystyle U_A(0,\gamma) \leq U_A\left(\frac{1}{N},\gamma\right) \\
\displaystyle U_B(0,\gamma) \leq U_B\left(0,\gamma+\frac{1}{N}\right) \\
\displaystyle U_B(0,\gamma) \leq U_B\left(0,\gamma-\frac{1}{N}\right)  
\end{array}\right.
\Leftrightarrow \left\{\begin{array}{@{\,}l@{}}
\displaystyle b \leq \frac{2b+a}{N} \\[1em]
\displaystyle b \leq (a+2b) (2\gamma+\frac{1}{N}) \\[1em]
\displaystyle b \geq (a+2b) (2\gamma-\frac{1}{N}) 
\end{array}\right.$$
$$ \Leftrightarrow \left\{\begin{array}{l}
\displaystyle 1 \leq \frac{4}{N} \\[1em]
\displaystyle 1 \leq 4 (2\gamma+\frac{1}{N}) \\[1em]
1 \geq 4 (2\gamma-1/N) 
\end{array}\right.
 \Leftrightarrow 
\begin{array}{@{}l}
\left\{\begin{array}{l}
\displaystyle N \leq 4 \\
\displaystyle \frac{N/8-1/2}{N} \leq \gamma
\end{array}\right.\\
\qquad \qquad \qquad \displaystyle \leq \frac{N/8+1/2}{N}
\end{array}$$

If $N \leq 3$ then $N/8+1/2 \leq 7/8 <1$ which cannot be obtained by the player otherwise than in $0$.
For $N=4$, the second inequality becomes $0 \leq \gamma \leq \frac{1}{4}$ which hence leads to the only non null Nash equilibrium.

We next study the potential equilibria of the form $(\fABC, 1)$.
Let $(\gamma, 1)$ be a Nash equilibrium. Then $ U_B(\gamma, 1) \leq U_B (\gamma, 1-1/N)$.
Then\\
$ \begin{array}{@{}l}
b\gamma + a+2b \leq b-b(1-1/N) + b\gamma + (a+2b) (1-1/N)^2 \\
 \Rightarrow a+2b \leq b/N + (a+2b) (1+1/N^2 -2/N) \\
\Rightarrow 0 \leq b + (a+2b) (1/N -2) \\
\Rightarrow 2a+3b \leq (a+2b)/N 
\Rightarrow N \leq 1/4.
\end{array}$
\end{proof}

\subsection{Boundary equilibria when \texorpdfstring{$a-2b>0$}{a-2b>0}.}

\begin{theorem}
\label{th:0x}
$(0, \alpha)$ and $(\alpha, 0)$ are Nash equilibria iff: 
$$ \frac{b}{a-2b} - \frac{1}{N}\frac{a+2b}{a-2b}
\leq \alpha \leq \frac{b}{2(a+2b)} + \frac{1}{2N}.$$ 

Further, there are no Nash equilibrium of the form $(A, 1)$.
\end{theorem}

\begin{proof}
We first focus on the Nash equilibria of the form $(0,A)$.
Since $U_A(.,\fBAC)$ and $U_B(\fABC,.)$ are convex, $(0, \gamma)$ is a Nash equilibrium iff
$$ 
\left\{\begin{array}{l}
\displaystyle U_A(0,\gamma) \leq U_A\left(\frac{1}{N},\gamma\right) \\
\displaystyle U_B(0,\gamma) \leq U_B\left(0,\gamma+\frac{1}{N}\right) \\
\displaystyle U_B(0,\gamma) \leq U_B\left(0,\gamma-\frac{1}{N}\right)  
\end{array}\right. $$

$$ \Leftrightarrow \hspace{-3pt} \left\{\begin{array}{@{\hspace{1pt}}l}
\displaystyle b \leq (a-2b)\gamma+\frac{2b+a}{N} \\[1em]
\displaystyle b \leq (a+2b) (2\gamma+\frac{1}{N}) \\[1em]
b \geq (a+2b) (2\gamma-\frac{1}{N}) 
\end{array}\right.
\hspace{-1em}\Leftrightarrow \left\{\begin{array}{@{\hspace{1pt}}l}
\displaystyle \gamma \geq \frac{bN-2b-a}{N(a-2b)} \\[1em]
\displaystyle \gamma \geq \frac{bN-a-2b}{2N(a+2b)} \\[1em]
\displaystyle \gamma \leq \frac{bN+a+2b}{2N(a+2b)}
\end{array}\right.
$$
But $ \frac{bN-2b-a}{N(a-2b)} \geq \frac{bN-a-2b}{2N(a+2b)} $ which concludes the proof.
 and hence $\frac{bN-a-2b}{2N(a+2b)} \leq \gamma \leq \frac{bN+a+2b}{2N(a+2b)} $

We now study the potential equilibria of the form $(A,1)$.
Let $(A, 1)$ be a Nash equilibrium. Then
$ \displaystyle U_B(A, 1) \leq U_B (A, 1-1/N)$.
Then
$$ \begin{array}{l}
-b + (a-2b) A + (a+2b) \leq -b(1-1/N) \\
\qquad  + (a-2b) A (1-1/N) + (a+2b) (1-1/N)^2 
\end{array}$$
$$ \Rightarrow 0 \leq b - (a-2b) A + (a+2b) (-2 + 1/N) $$
$$ \Rightarrow (a-2b) A \leq - 2a - 3b + (a+2b)/N \Rightarrow $$
$$ \Rightarrow 2a + 3b \leq (a-2b) A + 2a + 3b  \leq (a+2b)/N $$

But $ 2a+3b \leq (a+2b)/N \Rightarrow N \leq \frac{a+2b}{2a+3b} < 1$.
\end{proof}

\subsection{Proof of Theorem~\ref{th:big}.}

We first start by showing that there are at most $4$ interior Nash equilibria and that they are of the form: $(A,A)$,$(A+1,A)$,$(A,A+1)$,$(A+1,A+1)$.

\begin{proof}
Let $\fABC, \fBAC$ be a Nash equilibrium in the interior (i.e. $0 < \fABC < 1$ and $0 < \fBAC < 1$). Then $\fABC$ and $\fBAC$ are the (discrete) minimizers of $x \mapsto U_A(x, \fBAC)$ and $x \mapsto U_B(\fABC, x)$ respectively. Further:
$$ \left\{
\begin{array}{l}
\displaystyle \frac{\partial U_A}{\partial \fABC} = -b + (a-2b) \fBAC + 2 (2b+a) \fABC \\[1.2em]
\displaystyle \frac{\partial U_B}{\partial \fBAC} = -b + (a-2b) \fABC + 2 (a+2b) \fBAC
\end{array} \right.
$$

The optimum values are therefore respectively:
$$
x_A = \frac{b- \theta \fBAC}{\lambda} \text{ and }
x_B = \frac{b- \theta \fABC}{\lambda}
$$
with $\lambda = 2(2b+a)$ and $ \theta= a-2b$.
Therefore:
$$\left\{\begin{array}{l}
 x_A - \frac{1}{2N} \leq \fABC \leq x_A + \frac{1}{2N}\\[1em]
 x_B - \frac{1}{2N} \leq \fBAC \leq x_B + \frac{1}{2N}.
\end{array} \right.
$$ 
Hence
$$ \begin{array}{l} \displaystyle
\frac{b}{\lambda}-\frac{\theta}{\lambda} \left( \frac{b}{\lambda} - \frac{\theta}{\lambda} \fABC + \frac{1}{2N}\right) - \frac{1}{2N} \leq \fABC \leq \frac{1}{2N} \\[1em]
\displaystyle \hfill + \frac{b}{\lambda} -\frac{\theta}{\lambda} \left( \frac{b}{\lambda} - \frac{\theta}{\lambda} \fABC - \frac{1}{2N}\right) 
\end{array}
$$

Then
$$ 
\frac{b}{\lambda + \theta} - \frac{\lambda}{2N \left( \lambda - \theta\right)} \leq \fABC 
\leq \frac{\lambda}{2N \left( \lambda - \theta \right) } + \frac{b}{\lambda + \theta}
$$

Then
$ \displaystyle \frac{b}{\lambda + \theta} = \frac{b}{2b +3a}$, 
$ \displaystyle \frac{\lambda}{2N \left( \lambda - \theta \right) } = 
\frac{4b+2a}{2N \left( 6b+a \right) }$ and
$ \displaystyle \frac{\lambda}{2N \left( \lambda - \theta\right)} = 
\frac{ 2(a+2b)}{2N \left( 6b+a\right)}$, which gives

$$ 
\frac{b}{2b +3a} - \frac{a+2b}{N \left( 6b+a\right)} \leq \fABC 
\leq  \frac{2b+a}{N \left( 6b+a \right) } + \frac{b}{2b +3a}.
$$
Similarly, we have $$
\displaystyle \frac{b}{2b+3a} - \frac{(2b+a)}{N(6b+a) } \leq 
\fBAC \leq \frac{b}{2b+3a} + \frac{2b+a}{N (6b+a) }.$$
Note that $ \frac{1}{2} < \frac{2b+a}{6b+a} < 1$. Therefore there are either $1$ or $2$ possible values, which are identical for $\fABC$ and $\fBAC$. There are therefore $4$ possible equilibria.
\end{proof}

Now, the potential equilibria are of the form $(A, A)$, $(A, A+1)$, $(A+1, A)$ and $(A+1, A+1)$.
By symmetry, note that if $(A, A+1)$ is a Nash equilibrium, then $(A+1, A)$ also is. The following lemma reduces the number of combinations of equilibria:
\begin{lemma}
If $(A, A)$ is a Nash equilibrium then $(A+1, A+1)$ is not a Nash equilibrium.
\end{lemma}
\begin{proof}
Suppose that $(A, A)$ and $(A+1, A+1)$ are two Nash equilibria. Then
$U_A(A, A) \leq U_A(A+1, A)$ and $U_A(A+1, A+1) \leq U_A(A, A+1)$, which implies
$$ \left\{ \begin{array}{@{\,}l}
-bAN+(a-2b)A^2+(2b+a) A^2 \leq \\
\hfill -b(A+1)N+(a-2b)A(A+1)+(2b+a)(A+1)^2\\
-b(A+1)N+(a-2b)(A+1)^2+(2b+a) (A+1)^2 \leq \\
\hfill -bAN+(a-2b)A(A+1)+(2b+a)A^2\\
\end{array}\right.$$

$$ \Rightarrow \left\{ \begin{array}{l}
bN \leq (a-2b)A+(2b+a)(2A+1)\\
(a-2b)(A+1)+(2b+a) (2A+1) \leq bN \\
\end{array}\right.$$

$$
\Rightarrow
(a-2b)(A+1) \leq bN- (2b+a) (2A+1) \leq (a-2b)A
$$

Hence $ (a-2b)(A+1) \leq (a-2b)A$ and therefore $a-2b \leq 0$ which is impossible.
\end{proof}

Therefore the different possible combinations are mode 1, mode 2, mode 3-A and mode 3-B in Figure~\ref{fig:modes}).

We first start by the occurrence of mode 3-A:
\begin{lemma} \label{lem:1a}
Suppose that $a-2b>0$. Suppose that $(A, A)$ and $(A+1, A)$ are two Nash equilibria. Then
$$ A = \frac{bN-2b-a}{3a+2b}. $$
\end{lemma}
\begin{proof}
Suppose that $(A, A)$ and $(A+1, A)$ are two Nash equilibria. Then necessarily 
$ U_A(A, A) = U_A(A+1, A)$.
Hence
$$ \begin{array}{l}
-bAN+(a-2b)A^2+(2b+a) A^2 \\
\, = -b(A+1)N+(a-2b)A(A+1)+(2b+a) (A+1)^2 \end{array}$$
i.e.
$$ bN = (a-2b)A + (2b+a) (2A+1) 
\Rightarrow
bN-2b-a = (3a+2b) A 
$$
which leads to the conclusion.
\end{proof}

Hence, the system is in mode 3-A iff $bN-2b-a$ is divisible by $3a+2b$ or in other words, if $N$ is of the form $\frac{(3a+2b)K+2a}{b}$ for some integer $K$.

We then move on to Mode 3-B:
\begin{lemma} \label{lem:1b}
Suppose that $a-2b>0$. Suppose that $(A+1, A+1)$ and $(A+1, A)$ are two Nash equilibria. Then
$$ A = \frac{bN-2a}{3a+2b}.$$
\end{lemma}

\begin{proof}
Suppose that $(A+1, A+1)$ and $(A, A+1)$ are two Nash equilibria, then
$ U_1(A+1, A+1) = U_1(A, A+1)$. This implies
$$\begin{array}{l}
-Nb(A+1)+(a-2b)(A+1)^2+(2b+a) (A+1)^2 = \\
\hfill -NbA+(a-2b)A(A+1)+(2b+a) A^2
\end{array}
$$
$$ \Rightarrow (a-2b)(A+1) + (2b+a) (2A+1) = Nb $$
$$ \Rightarrow (3a + 2b) A = Nb- 2a $$
which concludes the proof.
\end{proof}

Hence, the system is in mode 3-B iff $bN-2a$ is divisible by $3a+2b$ or in other words, if $N$ is of the form $\displaystyle \frac{(3a+2b)K+2b+a}{b}$ for some integer $K$.

Finally, for Mode 2:
\begin{lemma}
Suppose that $a-2b>0$. Suppose that $(A, A+1)$ and $(A+1, A)$ are only two Nash equilibria. Then
$$ (3a+2b)A + 2b+a < bN < (3a+2b)A+ 2a.$$
\end{lemma}

\begin{proof}
Suppose that $(A, A+1)$ and $(A+1, A)$ are two Nash equilibria, then:
$$ \begin{array}{@{}l}
U_A(A, A+1) \leq U_A(A+1, A+1) \text{ and } \\
U_A(A+1, A) \leq U_A(A, A) 
\end{array}$$
ie
$$  \left\{ \begin{array}{l}
bN \leq (3a+2b)A+ 2a\\
(3a+2b)A + 2b+a \leq bN \\
\end{array}\right.
$$

The conclusion comes from Lemma~\ref{lem:1a} and \ref{lem:1b}, since neither $(A,A)$ nor $(A+1, A+1)$ are Nash equilibria.
\end{proof}

Finally the system is in mode 1 if it is not in any over modes. One can then check that the boundary cases found in 
Theorem~\ref{th:0x} corresponds to the case where $A=0$ which concludes the proof.

\end{document}